\documentclass[12pt,reqno]{amsart}
\usepackage[utf8]{inputenc}
\usepackage{amsmath,amssymb,amsthm,hyperref,array,xcolor,mathtools,multicol,subcaption,verbatim,microtype}
\hypersetup{
pdftitle={Mechanism Design by a Politician}, 
pdfauthor={Giovanni Valvassori Bolgè}, 
pdfkeywords={Legislative Coalitions, Public Goods, Mechanism Design, Reservation Utility}, 
pdfcreator={LaTeX with hyperref package}, 
pdfproducer={LaTeX},
colorlinks=true,
linkcolor=[rgb]{0,0,0.6}, 
citecolor=[rgb]{0,0,0.6},  
urlcolor=[rgb]{0,0,0.6},   
}

\usepackage[normalem]{ulem}
\usepackage[pdftex]{graphicx}
\usepackage{fullpage}
\usepackage{cleveref}

\usepackage{tikz}
\usepackage{tikz-3dplot}
\usetikzlibrary{patterns,patterns.meta,arrows.meta, calc}

\usepackage{enumitem}
\setlist[itemize]{nosep}
\setlist[enumerate]{nosep}

\usepackage{natbib}

\newtheorem*{theorem*}{Theorem}
\newtheorem{theorem}{Theorem}
\newtheorem{lemma}{Lemma} 
 
\newtheorem{definition}{Definition} 
 
\newtheorem{assumption}{Assumption}

\newtheorem{critique}{Critique}
\newtheorem{property}{Property}
\newtheorem{claim}{Claim}

\theoremstyle{remark}

\DeclareMathOperator*{\argmax}{arg\,max}

\usepackage{accents}

\crefname{equation}{}{}

\begin{document}
	
\title{Equality of Opportunity and Opportunity Pluralism*}
\thanks{*This work is a generalization of the Master's dissertation I defended in September 2019 under the supervision of Professor François Maniquet at UCLouvain. I am heavily indebted to Professor Maniquet for his many suggestions and comments. }
	 
\author{Giovanni Valvassori Bolgè\textsuperscript{\textdagger}}\thanks{\textsuperscript{\textdagger}University of Fribourg, Department of Economics. \textit{Email: \href{mailto:giovanni.valvassoribolge@unifr.ch}{giovanni.valvassoribolge@unifr.ch}}}
	
\date{\today}

\begin{abstract} 

	This paper seeks to explore the potential trade-off arising between the theories of \textit{Equality of Opportunity} and \textit{Opportunity Pluralism}. Whereas the first theory has received much attention in the literature on Welfare Economics, the second one has only recently been introduced with the publication of the book by Joseph Fishkin, \textit{Bottlenecks: A New Theory of Equal Opportunity}. After arguing extensively that any notion of human flourishing is incompatible with traditional theories of \textit{Equality of Opportunity}, the author proposes an alternative theory squarely based on a broad notion of human development. This paper seeks to formalize the argument made in this book through the lens of economic theory. My analysis suggests that traditional theories of \textit{Equality of Opportunity} are not incompatible with  \textit{Opportunity Pluralism}.

	
\end{abstract}

\maketitle

\newpage

\section{Introduction}\label{sec1}

The theoretical foundations of Equality of Opportunity (EOp, henceforth) might be traced back to the pioneering contributions of the political philosopher John Rawls (\cite{rawls1958}, \cite{rawls1971}). By crucially integrating personal responsibility with egalitarianism, Rawls paved the way for a large literature in Political Philosophy; soon, other major contributions followed, most notably, from \cite{sen1980} and \cite{dworkin1981}.

More recently, the theory of EOp has found its way into normative economics, owing to the work of John Roemer (\cite{roemer1993}, \cite{roemer2000}). Ever since, there has been a strong development both in the theoretical underpinnings and in the empirical applications of EOp.\footnote{For a detailed survey, see \cite{roemer2016}.} According to \cite{maniquet2017}, these developments can be rationalized in three different approaches within Economics.

The first approach is constituted by the theory of \textit{fair resource allocation}. Starting from the fact that individuals are heterogeneous along several characteristics, economic theorists have focused on devising resource allocation mechanisms which are compatible with different ethical principles of social justice.

Building on the work of John Roemer, \cite{fleurbaey2011} propose to divide the relevant individual characteristics into two sets: the \textit{compensation characteristics} are those whose inequality between individuals is deemed as ethically objectionable (e.g. innate abilities); the \textit{responsibility characteristics} are constituted by those individual traits for which individuals should be held accountable (\textit{effort} variables; e.g. schooling). Several theories have been advanced to identify the cut between compensation and responsibility variables. Although they differ in the particular ethical principles which inspire them, they all agree that society should redistribute resources in order to compensate the inequality deriving from the first set of characteristics, while remaining neutral with respect to the inequality stemming from the second set.\footnote{For a thorough discussion, see, \textit{inter alia}, \cite{fleurbaey2008}.}

The second approach relates to the theory of income inequality measurement. Following the groundbreaking contributions of \cite{pigou1912} and \cite{dalton1920}, economists have been able to decompose the income distribution as the sum of different components, some of which calling for compensation and others for neutrality by society (see \cite{peragine2002} and \cite{bourguignon2007}).

Measuring the extent to which the income distributions differ from one of perfect equality due to ethically objectionable characteristics is the basis of many different empirical studies, most notably the World Bank \textit{Report on Equity and Development} (\cite{worldbank2006}) as well as the thorough empirical analysis for Latin America in \cite{paesdebarros2009}.

The third approach has been devised within the theory of inequality of opportunity measurement. By considering the set of opportunities open to each individual, economists have focused their analysis not only on the elements of the sets eventually chosen by individuals, but also on the cardinality of the sets themselves (e.g. \cite{thomson2011}; \cite{kranic1998}).

More recently, an insightful critique of the theoretical foundations of EOp has been published by the political theorist Joseph Fishkin in the book \textit{Bottlenecks: A New Theory of Equal Opportunity} (\cite{fishkin2014}).

The author claims to advance a radically new way of conceptualizing  EOp --- which he refers to as \textit{Opportunity Pluralism} --- building a novel theory squarely on an account of human development.

Unlike traditional theories of EOp which --- the author argues --- have narrowly focussed on the \textit{instrumental} role of opportunities, \textit{Opportunity Pluralism} ascribes a central role to the \textit{developmental} features of opportunities, namely the broadly defined way in which opportunities shape individual development.

In particular, the first condition of an ideal society satisfying \textit{Opportunity Pluralism} is the presence of a ``plurality of values and goals'', which I refer to as \textit{Plurality of Opportunity} (POp, henceforth). In the words of \cite{fishkin2014},
\begin{quote}
    ``People in this society hold diverse conceptions of the good, of what kinds of lives and forms of flourishing they value, and of what specific goods and roles they want to pursue [...].''\footnote{\cite{fishkin2014}, p. 135.}
\end{quote}
Importantly, Fishkin argues that POp is incompatible with traditional EOp theories. His line of reasoning suggests the existence of a trade-off faced by society: either society intervenes to equalize opportunities on a chosen scale in the traditional sense or it promotes \textit{Opportunity Pluralism}, by removing the different types of \textit{bottlenecks} which constrain the individual development of its citizens.

The present paper aims at evaluating two salient critiques moved by Joseph Fishkin to EOp theories through the lens of basic economic reasoning. First, the purpose is to evaluate his claim regarding the impossibility of incorporating the developmental role of opportunities into the standard framework of EOp adopted by economists. Second, I investigate whether this conceptualization of EOp is incompatible with POp, and therefore with \textit{Opportunity Pluralism}, as Fishkin argues.

In order to do this, I build a simple two-period, two-sector model in an economy populated by agents with heterogeneous productive skills over the two sectors. An egalitarian government seeks to equalize opportunities among agents through educational investments in either of two schools, where skills in each sector are nurtured.

The preferences of the agents are uncertain and endogenous to educational investments. Agents do not know their preferences over the two sectors when educational investments take place. Moreover, the realization of preferences will depend on the bundles of opportunities received.

I show that a \textit{pluralistic opportunity structure}, namely a society with two schools, yields a higher social welfare than a \textit{unitary opportunity structure}, a society with a single school.

Therefore, my analysis shows that EOp \textit{implies} a pluralistic opportunity structure, the notion at the heart of POp, and consequently of \textit{Opportunity Pluralism} as well. As a result of this, I claim that the theories of EOp and \textit{Opportunity Pluralism} are not incompatible.

I structure the discussion as follows. In \Cref{sec2} and \Cref{sec3}, I briefly describe the main lines of reasoning of Fishkin's contribution, articulating them into two distinct critiques to EOp. \Cref{sec4} lays out the theoretical framework through which I frame those critiques. \Cref{sec:resourceq} and \Cref{sec:utilityeq} consider the nature and evaluation of the different education policies. \Cref{sec:conclusions} concludes.

\section{Opportunities and Human Development}
\label{sec2}

Fishkin envisions human flourishing as a complex, multi-stage process in which three sets of factors interact with each other: the environment, comprehensive of \textit{developmental opportunities}; the person, as constituted by his/her capacities, traits and genetic heritage; his/her ambitions, goals and efforts, namely the ``multidimensional path" a person is embarking on as a result of the previous layer of interactions between the environment and his/her personal and social background.

Moreover, this process is iterative. Some unusual capacities at an early age will affect the range of developmental opportunities provided by one's family and environment; importantly, also the converse holds. In the author's own words:
\begin{quote}
    ``... our ambitions, goals and efforts do not emerge fully formed from the ether, but are instead products of our lived experience; they, in turn, influence other aspects of the processes by which we develop traits and capacities, convince others to recognize our capacities, prove ``our" merit, and secure jobs and other social roles. Our decisions about how to direct our efforts are in part a function of the choice set of paths and options we see before us at each stage - as well as our own conclusions, mediated by the conclusions of others, about our merit and potential."\footnote{Ibid. p.115.}
\end{quote}
According to Fishkin, if we embrace this perspective with regard to human development, the theoretical foundations of any framework for \textit{equalizing} opportunities begin to crumble.
In fact, the author argues:
\begin{quote}
    ``Equalizing developmental opportunities means that we ought to arrange the different bundles of developmental opportunities that children receive on some sort of scale; then we aim for state of affairs in which everyone has a bundle of opportunities of equal value, or as close as possible to equal value, on that scale."\footnote{Ibid. p.116.}
\end{quote}
This line of reasoning is fatally flawed because:
\begin{quote}
    ``... the more carefully we think through the iterative interactions that make up human development, the less clear it becomes what `equal' developmental opportunities could possibly mean."\footnote{Ibid. p.116.}
\end{quote}
As a result of this, 
\begin{quote}
    ``... it becomes impossible either to arrange the different bundles of opportunities on a single scale or to identify any one set of developmental opportunities that could function as a fair baseline of equality for everyone."\footnote{Ibid. p.116.}
\end{quote}
This difficulty, the \textit{common scale problem}, is due to two other problems: the \textit{incommensurability} of bundles of opportunity and the \textit{endogeneity} of preferences.

By incommensurability, Fishkin means ``the problem that different developmental opportunities are valuable for different reasons, and some of those reasons are incommensurable with one another."\footnote{Ibid. p.119.}

This problem is compounded by the one of the endogeneity of preferences. This builds directly on the model of iterative human development discussed above. The layers of interactions between the environment, the genetic heritage and the set of personal traits all contribute to shaping one's values, goals and ambitions. As a consequence, 
\begin{quote}
    ``... the endogeneity of our preferences and values --- specifically, their dependence on our developmental opportunities and experiences --- can make it difficult or impossible to say \textit{ex ante} which bundles of developmental opportunities are best for a person."\footnote{Ibid. p.124.}
\end{quote}
I summarize this discussion as the first fundamental critique of Fishkin:
\begin{critique} \label{critique1}
It is impossible to build a metric of opportunities so as to take their developmental role into account.
\end{critique}

\section{The Role of the Opportunity Structure and the Trade-off}
\label{sec3}
The shape of a society's \textit{opportunity structure} --- Fishkin argues --- greatly affects the degree to which the common scale problem binds. In the author's own words,
\begin{quote}
    ``... the unitary or pluralistic character of the opportunity structure has an interesting effect on the commensurability of different bundles of opportunities. The more unitary the opportunity structure, the more often we will be able to say objectively that a given bundle of opportunities is better or worse than another."\footnote{Ibid. p.126.}
\end{quote}
In the extreme case of the ``big test society", to Fishkin refers to many times throughout the book --- namely, a fictitious society in which access to all sought-after jobs and positions in adulthood depends on the score on a single test undertaken at age 18 --- only those skills essential to passing the test are crucial. In this regard, deciding which bundles of developmental opportunities are better or worse becomes much easier for an ideal, benevolent planner. However, such a planner should advocate exactly the opposite, since a unitary opportunity structure exacerbates all the intrinsic problems of equal opportunity which Fishkin discusses in Part I of his book.

In the limiting case of the big test society, in fact,  \textit{the problem of the family} binds even more tightly because parents have every incentive to maximize the development opportunities for their children to score well on the test. Hence, more affluent families will have an unfair advantage over poorer ones.

This, in turn,  triggers \textit{the problem of the merit} and of the \textit{starting gate}: in a fascinating treatment of the subject, Fishkin argues extensively that there will be ``no fair way to equalize opportunities among people whose developmental opportunities were unequal."\footnote{Ibid. p.128.}

Lastly, a unitary opportunity structure is the worst possible society as far as \textit{the problem of the individuality} is concerned: there is only one goal everyone agrees about; life is really a race. Therefore, there is no space for ``anyone to develop in other directions, devote themselves to other goals, or to carve out new paths of their own."\footnote{Ibid. p.128.}

All this discussion leads the author to formulate the following bold conclusion:
\begin{quote}
    ``... even though a unitary opportunity structure might help make the problem of equalizing opportunity \textit{conceptually} more tractable, its \textit{substantive} effect is just the opposite: it makes all the things that made us value equal opportunity in the first place that much harder to achieve."\footnote{Ibid. p.128.}
\end{quote}
This constitutes the second fundamental critique moved by Fishkin to traditional EOp theories:
\begin{critique}
    Equalizing opportunities within a pluralistic opportunity structure is intractable.
\end{critique}
Political philosophers and economic theorists seem to be stuck: either they focus on a unitary opportunity structure --- a ``unidimensional society" which severely hinders human flourishing  --- or defining EOp within a pluralistic opportunity structure becomes conceptually intractable, leading to the necessity of elaborating the new theory of \textit{Opportunity Pluralism}.

In fact, the two critiques point to the existence of a fundamental trade-off between EOp and \textit{Opportunity Pluralism}.

\section{The Theoretical Framework}
\label{sec4}
Since the language and the tools of economists and political philosophers differ, I make two fundamental remarks about the way I frame the theoretical claims moved by Fishkin through the lens of economic theory. The first one relates to the respect of preferences. The second one deals with the way I model human development. 

One of the normative values which provide the basis of economists' toolbox is the respect of preferences. People are autonomous moral agents who have the right to choose their own notion of good life, without it being imposed from an external authority. Therefore, an individual's choices can reasonably be interpreted as revealing information about the nature of this notion. Under this assumption, economists construct metrics of well-being to measure different bundles of opportunities. The key point is that this is not incompatible with viewing opportunities as a crucial factor in shaping one's values, goals and ambitions. Even if economists held the very same view as Fishkin, it would still be justifiable to construct measures of individual well-being, building on what revealed (and well-informed) preferences induce individuals to choose.\footnote{For a thorough discussion, see \cite{fleurbaey2011} and \cite{maniquet2017}.}

It seems that Fishkin's argument that it is impossible to come up with a measure of well-being for different bundles of opportunities is actually incompatible with the respect of preferences.\footnote{See the discussion leading up to \Cref{critique1}.} In this regard, I depart from Fishkin's view and I adhere to the economic principle of respecting individual preferences. This enables to adopt utility theory and model an economy in which bundles of opportunities can be arranged according to a common scale.

Second, I model human development by assuming that preferences are both uncertain and endogenous. In this way, I try to follow Fishkin's claim that one's preferences evolves over time, depending on those multi-stage interactions between one's genetic endowment, the environment and the opportunities provided in a way that cannot be thoroughly predicted \textit{ex ante}, by talents at birth.

\subsection{The Model}
The economy is populated by a continuum of individuals, and an egalitarian government seeks to equalize opportunities among them. There are two occupations, \textit{a} and \textit{b}, each of which requires a different skill. Individuals are heterogeneous with respect to the talents (or marginal productivities) in the two sectors; individual $i$ 's talents are represented by a vector $\boldsymbol{t}^{i}= (t_{a}^{i} , t_{b}^{i}) \in  \Re_{+}^{2}$. 
\begin{definition}
    A distribution of talents is \textit{pluralistic} if there are at least two individuals --- $i$ and $j$ --- such that $t_{a}^{i}>t_{b}^{i}$ and $t_{b}^{j}>t_{a}^{j}$.
\end{definition}
As a minimal requirement of heterogeneity among individuals, we make the following
\begin{assumption}
The distribution of talents is \textit{pluralistic}.
\end{assumption}
Bundles of opportunities are represented by educational investments in either of two schools, $s_a$ and $s_b$, depending on whether the marginal productivity in occupation $a$ or $b$ is nurtured. 
Individuals are also heterogeneous with respect to a taste parameter $\theta$, a random variable with distribution $F$ and density function $f$. This parameter represents the preference of each individual for occupation $a$ relative to $b$. As it will be explained below, this preference parameter captures the uncertainty and the endogeneity of preferences.
\paragraph{Timing}There are two periods, $t=1,2$.
At $t=1$, the egalitarian government observes the talents of the individuals and assigns them to either of two schools, $s_a$ and $s_b$, where, depending on the amount of educational expenditure received, they will increase their productivity in either sector \textit{a} or \textit{b}, respectively.\footnote{I assume that, within the same school --- $a$ or $b$ --- the amount of educational expenditure is individual-specific.}

Then, at $t=2$, individuals observe the realization of $\theta$ and decide whether to work in sector \textit{a} or \textit{b}, depending on the value of $\theta$ and the wages that they will earn in occupation \textit{a} or \textit{b}.

\paragraph{The Technology of Educational Investments}
The wages each individual would earn on the labour market at $t=2$ are a function of his/her talents and the amount of educational expenditure received.
Within each occupation, denote $\textit{w}_{j}^{i} : \Re^2_{+} \rightarrow  \Re_{+}$ as the wage earned by individual \textit{i} endowed with talent $t_j^i$, after receiving an amount $e_j^i \in \Re_+$ of educational expenditure.
For simplicity, additive separability between talent and expenditure within each sector is assumed; i.e. for each $j \in (a,b)$,\footnote{All results go through for more general technologies of educational investments in which talent and educational investments are complementary.}
\begin{equation*}
   \textit{w}_{j}^{i} = \textit{t}_{j}^{i} + \textit{e}_{j}^{i} \ .
\end{equation*}
Each individual receives educational inputs in only one of the two sectors. In the sector for which s/he is not educated, the wage at $t=2$ would be equal to his/her talent (or marginal productivity).

\paragraph{Uncertainty of Preferences}
During the schooling process, the individual does not know what his or her preferences will be at $t=2$. I follow \cite{arrow1995} in modelling uncertainty about one's own preferences. Two features of this framework seem to fit particularly well within Fishkin's account of opportunities.

First, I adopt a view of opportunity as a component of well-being or, in Arrow's terms, a view of \textit{freedom as autonomy}: freedom of choice --- represented by \textit{opportunity sets} ---  itself constitutes part of well-being, independently of the actual choice made by the individual.

Secondly, this freedom of choice is represented by uncertainty about one's future preferences; i.e. \textit{autonomy} is interpreted as freedom to choose one's own preferences.
The conceptual framework devised by \cite{arrow1995} is as follows.\\
Each individual is endowed with an opportunity set $A^i$ whose generic element is denoted by $x^i$. Individual \textit{i}'s  utility at $t=1$ is represented by a utility function $U(x^i, \theta )$, where $\theta$ is a random variable with a (known) distribution which governs how the preferences of the individual will be at $t=2$.\\
For any given opportunity set $A^i$, the achieved utility of the individual would be:
\begin{equation*}
    \mathcal{P}(\theta, A^i) = [ \max   U(x^i, \theta ) \ | \ x^i \in A^i] \ .
\end{equation*}
However, since $\theta$ is not known at $t=1$ when $A^i$ is assigned, the expected value of this payoff is:
\begin{equation*}
    \mathcal{V}(A^i) = \mathbb{E}_{\theta} [ \max  U(x^i, \theta ) \ | \ x^i \in A^i] = \mathbb{E}_{\theta} [\mathcal{P}(\theta, A^i)] \ .
\end{equation*}
This function, $ \mathcal{V}(A^i)$ is referred to in \cite{arrow1995} as \textit{freedom evaluation function}.
It has two basic properties.
\begin{property} \label{property1}
If the opportunity set $A^i$ is included in $B^i$, then, for each $\theta$, 
    \begin{equation*}
          \max  [  U(x^i, \theta) \ |\ x^i  \in A^i]  \leq  \max  [  U(x^i, \theta) \ | \ x^i \in B^i]
    \end{equation*}
    or, equivalently,
    \begin{equation*}
        \mathcal{V}(A^i) \leq \mathcal{V}(B^i), \ \text{if} \ A^i \subseteq B^i .
    \end{equation*}
    \end{property}
    \begin{property}
    Only those elements which maximize $U(x^i, \theta)$ for some $\theta$ matter. Formally, let\\
    \begin{equation*}
    \xi( \theta, A^i) = \argmax [U(x^i, \theta) \ | \ x^i \in A^i]
    \end{equation*}
    and let\\
    \begin{equation*}
        E(A^i) = \text{range}_{\theta} \xi( \theta, A^i).
    \end{equation*}
    \medskip\\
    Then,
    \begin{equation*}
         \mathcal{V}(A^i) =  \mathcal{V}(B^i) \iff E(A^i) = E(B^i) \ .
    \end{equation*} 
    \end{property}

Thus, adding an element to the opportunity set which is not better than all existing elements for at least one possible utility function parameter $\theta$ does not increase freedom nor, consequently, utility.

\subsection{The Opportunity Sets}
I assume that preferences are represented by the following utility function
\begin{equation} \label{eq:utility}
    U(l_a^i, l_b^i, \theta )= \theta l_a^i + (1-\theta) l_b^i \ ,
    \end{equation}
    where $l_{a}^i$ and $l_{b}^i$ denote occupations $a$ and $b$ for individual $i$.\footnote{This choice is motivated by the assumption that each individual can work in a single occupation, not in both. Usually, labour time enters negatively in the utility function as disutility from labour, or positively as leisure. Instead, in this model $l_{a}^{i}$ and $l_{b}^{i}$ represent two normal goods (e.g., ``being a mathematician'' and ``being a musician'') or simply the extensive margin of labour in each occupation. It is straightforward to generalize the model to include consumption and labour time entering negatively into the utility function (the intensive margin).}

I assume that each individual $i$ is endowed with an opportunity set $A^i$ with constant transformation rates given by:
\begin{equation} \label{eq:oppsets}
    (l_{a}^{i} / w_{a}^{i}) + (l_{b}^{i} / w_{b}^{i}) \leq 1 \ .
\end{equation}
This is displayed in \Cref{fig:oppsets}.
\begin{figure}[ht]
\centering
\begin{tikzpicture}[scale=0.8]
\draw [<->] (0,6) node[left]{$l_{b}^{i}$}  -- (0,0) -- (6,0) node[below]{$l_{a}^{i}$};
\draw [very thick] (0,4)node[left]{$w_{b}^{i}$} -- (4,0) node[below]{$w_{a}^{i}$};
\end{tikzpicture}\\
\caption{The Opportunity Sets} \label{fig:oppsets}
\label{fig:oppsets}
\end{figure}
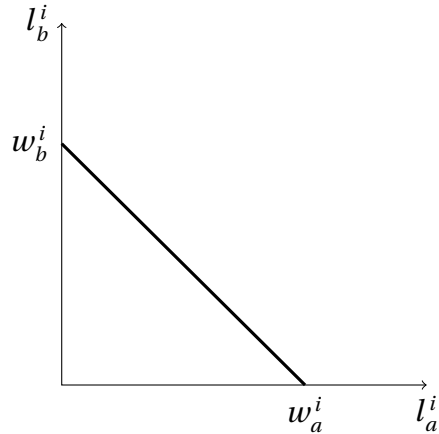\\
Maximizing (\ref{eq:utility}) over (\ref{eq:oppsets}) yields the indirect utility:
\begin{equation*}
    \mathcal{P}(w_{a}^i, w_{b}^i;\theta)= \max \big\{\theta \textit{w}_{a}^i , (1-\theta)\textit{w}_{b}^i\big\} \ .
\end{equation*}
The \textit{freedom evaluation function} is defined analogously as
\begin{equation} \label{eq:fefunction}
    \mathcal{V}(w_{a}^i, w_{b}^i) = \mathbb{E}_\theta [ \max \big\{\theta \textit{w}_{a}^i , (1-\theta)\textit{w}_{b}^i\big\}] \ .
\end{equation}
At $t=2$, when $\theta$ is realized, individual \textit{i} will decide whether to work in occupation \textit{a} or \textit{b} depending on the inequality:
\begin{equation*}
    \theta \textit{w}_{a}^i >  (1-\theta)\textit{w}_{b}^i \ .
\end{equation*}

\paragraph{Endogeneity of Preferences}
Preferences are endogenous in the following sense. The distribution of the taste parameter $\theta$ is allowed to depend on wages in the two sectors, and hence on the educational investments received by individuals.

Formally, denote $F(\theta; w_a^i, w_b^i)$ the distribution of the taste parameter, and $f(\theta; w_a^i, w_b^i)$ the density function. I make the following 
\begin{assumption} \label{assumption2}
    The density $f$ satisfies, $\forall (w_a^i,w_b^i) \in \Re^2_+$,
  \begin{enumerate}
  \item[(i)] $f(\theta; w_a^i, w_b^i)=f(1-\theta;w_b^i,w_a^i)$ \ ;
  \item[(ii)] If $w_a^i=w_b^i$, then $f$ is symmetric on $[0,1]$ \ ;
  \item[(iii)] If $w_a^i \geq (\leq) \ w_b^i$, $f$ is skewed to the left (right).
  \end{enumerate}

\end{assumption}
Condition \textit{(i)} simply means that relabelling the two sectors and the taste parameter does not affect the (expected) utility of individuals.

The other two conditions provide a very minimal constraint on the ``direction" of the endogeneity of preferences. To see why this is the case, recall the \textit{freedom evaluation function} in (\ref{eq:fefunction}), and suppose that $w_a^i=w_b^i$.

Condition \textit{(ii)} says that it is equally likely to observe a realization of $\theta$ above or below $\frac{1}{2}$; in other words, it is \textit{ex ante} equally likely for $i$ to choose occupation $a$ or $b$ at time $t=2$.  

If $w_a^i \geq w_b^i$, then by Condition \textit{(iii)}, $f$ puts at least as much mass to the right of $\frac{1}{2}$ as it puts on its left. Intuitively, this means that individual $i$ is at least as likely to observe a realization of $\theta$ which will induce him/her to choose occupation $a$ as occupation $b$. 

Moreover, \Cref{assumption2} ensures that the \textit{freedom evaluation function} satisfies two important properties. 
\begin{property} \label{property:symmetry}
$\mathcal{V}$ is symmetric.
\end{property}
In addition to this, because of the specification of opportunity sets in (\ref{eq:oppsets}), we also have the following
\begin{property} \label{property:quasiconvexity}
    $\mathcal{V}$ is quasiconvex.
\end{property}
These properties will play an important role in the characterisation of the education policies in \Cref{sec:resourceq} and \Cref{sec:utilityeq}. Their proofs are in \Cref{sec:Appendix}.
\paragraph{The Egalitarian Government}
The government is egalitarian; hence, it cares only about the least well-off. Moreover, the government has to make educational investments at $t=1$, but individuals will learn their preferences only at $t=2$. Therefore, the social welfare function is given by
\begin{equation} \label{SWF}
\mathcal{W} =  \max_{e_a,e_b} \min_i \mathcal{V}(w_a^i,w_b^i) = \max_{e_a,e_b} \min_i \mathbb{E}_{\theta} \Big[\max\big\{\theta w_{a}^i, (1-\theta)w_{b}^i\big\}\Big]
\end{equation}
Finally, I assume that the budget $\Bar{R}$ allocated to educational investment is large, but finite.
\begin{assumption}
\begin{equation*}
 0  \ll \Bar{R} < \infty \ .
\end{equation*}
\end{assumption}
Therefore, all education policies presented in \Cref{sec:resourceq} and \Cref{sec:utilityeq} can be financed.

\section{The Education Policies: Equalization of Resources} \label{sec:resourceq}
Suppose that the egalitarian government seeks to equalize wages among individuals. Given that there is fixed educational budget, $\bar{R}$, the government faces three alternatives. Either it equalizes wages within sector \textit{a}, or sector \textit{b}, or between both.
\begin{definition}
    The \textit{one-school policy} in sector $j=\{a,b\}$ equalizes wages in sector $j$ by investing the per-capita amount $e_{j}^{i} = |{w}^* - t_{j}^{i}|$, for each individual $i$ and some $w^* \in \Re_+$.
\end{definition}
The third possibility is sending individuals to the two different schools, depending on whether they have more talent for sector \textit{a} or sector \textit{b}.

\begin{definition}
    The \textit{two-school policy} equalizes wages in each sector $j=a,b$ by investing the per-capita amount $e_j^i=|\Tilde{w} -  \max_j \left\{t_{a}^i, t_{b}^i \right\}| $, for each individual $i$ and some $\Tilde{w} \in \Re_+$.
\end{definition}
It is straightforward to derive the following lemma.
\begin{lemma} \label{lemma:resorceq}
For any pluralistic distribution of talents and a given budget $\Bar{R}$, the two-school policy enables to equalize wages at a higher level than either of the one-school policies within each sector.
\end{lemma}

\subsection{The Evaluation of the Three Education Policies}

It is straightforward to show that an egalitarian government will choose the two-school policy.
To illustrate the argument, consider the distribution of talents given by $t_a^i +t_b^i=1$ in \Cref{Fig:eqwages}.
\begin{figure}[ht]
\centering
\begin{tikzpicture}[scale=0.8]
\draw [<->] (0,8) node[left]{$t_{b} + e_{b}$}  -- (0,0) -- (8,0) node[below]{$t_{a} + e_{a}$};
\draw [very thick] (0,4.8)node[left]{1} -- (4.8,0) node[below]{1};
\draw [dashed] (0,0)node[left]{0} -- (4.8,4.8);
\draw [draw=blue, very thick] (0,4.8) -- (4.8,4.8);
\draw [draw=blue, very thick] (4.8,0) -- (4.8,4.8);
\draw[thin] (2.4, 2.4) -- (2.4, 6);
\draw[thin] (2.4, 2.4) -- (6, 2.4);
\draw [draw=red, very thick] (0, 6) node[left]{1.25} -- (2.4, 6);
\draw [draw=red, very thick] (6,0) node[below]{1.25} -- (6, 2.4);
\end{tikzpicture}\\
\caption{Evaluating the three policies \ .} \label{Fig:eqwages}
\end{figure}
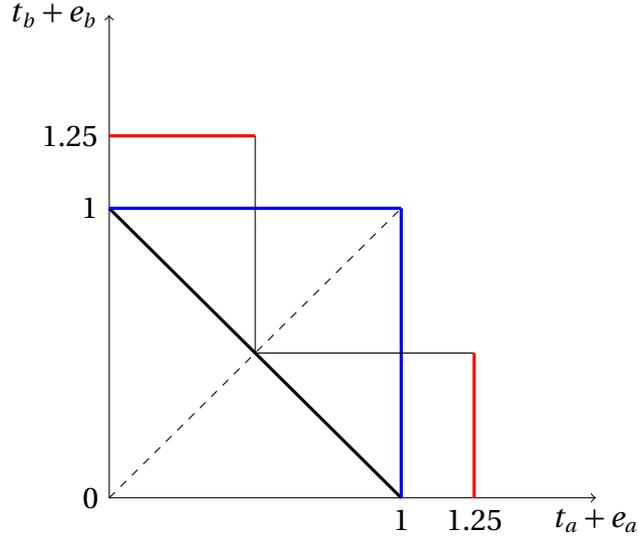

If wages were equalized in sector $a$, then the most talented individual in the distribution, located at (1,0), would become \textit{ex post} the least well-off: everyone will have the same wage in sector $a$ as him/her, but would have strictly higher wage in sector $b$. This is represented by the horizontal blue segment.

The case for sector $b$ would be analogous, but in that case, the individual located at (0,1) would be the least well-off \textit{ex post}. This is given by the vertical blue segment.
Considering these two individuals, either of the one-school policies would bring one of the two at (1,1) while leaving the other one at his/her own original coordinates.

However, the two-school policy --- given by the two red segments --- would yield an expected indirect utility level to both individuals which would be higher than their worst-case scenario. This follows from \Cref{property1}. \Cref{lemma:resorceq} generalises the result to any pluralistic distribution of talents. Thus, we have the following theorem.
\begin{theorem}
The two-school policy yields a higher social welfare than any of the one-school policies.
\end{theorem}

\section{The Education Policies: Equalization of Expected Utilities} \label{sec:utilityeq}
The above set of education policies might seem highly questionable on at least two levels.
On the one hand, it was assumed that everyone could be educated to reach any marginal productivity level. This assumption --- sometimes called \textit{educability} in the literature (e.g., see \cite{kranic1998}) ---  might well be deemed as highly unrealistic, even after weighting the education cost of low-talent individuals more heavily.

In fact, it would not be unreasonable to assume that some individuals might never reach the marginal productivity levels of the high-talent ones, not even with an arbitrarily large amount of educational expenditure.

On the other hand, the education policies presented above imply that individuals with equal taste parameter $\theta$ are likely end up with a different utility level, not only \textit{ex ante} (at $t=1$), but also \textit{ex post} (at $t=2$).

In the literature regarding the axiomatisation  of EOp, one of the commonly-used axioms is referred to as \textit{Equal Welfare for Equal Preference} (e.g., \cite{fleurbaey2008}). Under the assumption that skills and talents are individual characteristics that call for compensation, whereas taste differences should not, the axiom requires that the welfare level of individuals with the same taste parameter be the same.
This axiom is likely to be violated by the education policies discussed above.

Therefore, I now assume that the government aims at investing in individual education in such a way as to bring everybody to the same \textit{expected} indirect utility level curve. Denote $\mathcal{V}^{-1}(k^*)$ the highest expected indirect utility level curve that all individuals can reach, given the available educational budget $\Bar{R}$.\footnote{Formally, this means defining the indifference map  $\{\mathcal{V}^{-1}(k): k \in \Re_+ \} $, with $\mathcal{V}^{-1}(k)=\{(w_{a},w_{b}) \in \Re^2_{+}: \mathcal{V}(w_{a},w_{b})=k\}$. Then, $k^*$ is the unique value of the level curve such that the total cost of the policy is equal to $\Bar{R}$.}
\begin{definition}
    The \textit{one-school policy} in sector $j=\{a,b\}$ equalizes expected utilities in sector $j$ by investing the per-capita amount $e_{j}^{i} =|\mathcal{V}^{-1}(k^*)-t_{j}^{i}|$, for each individual $i$.
\end{definition}
\begin{definition}
    The two-school policy equalizes utilities in each sector $j=a,b$ by investing the per-capita amount $e_j^i=|\mathcal{V}^{-1}(k^*)- \max_j \{t_{a}^i, t_{b}^i\}| $ for each individual $i$.
\end{definition}

\subsection{The Evaluation of The Education Policies}
Consider again the case in which the distribution of talents is given by $t_a^i+t_b^i=1$. It is straightforward to show that the two-school policy is yields a higher social welfare, for \textit{any} set of welfare weights in the social welfare function (\ref{SWF}).

To see this, suppose that the highest level curve allowed by the educational budget in both one-school policies is the one in blue in \Cref{fig:utilityeq}.
\begin{figure}[ht]
\centering
\begin{tikzpicture}[scale=0.8]
\draw [<->] (0,8) node[left]{$t_{b} + e_{b}$}  -- (0,0) -- (8,0) node[below]{$t_{a} + e_{a}$};
\draw [very thick] (0,5.96)node[left]{1} -- (5.96,0) node[below]{1};
\draw [dashed] (0,0)node[left]{0} -- (4,4);
\draw [draw=red, thick] (3,3) -- (3,4.88);
\draw[draw=red, thick] (3,3) -- (4.88, 3);
\draw [draw=blue, thick] (0,5.96) to [out=353,in=99] (5.96,0); 
\draw [draw=red, thick] (3,4.88) to [out=325,in=125] (4.88, 3)  ;
\end{tikzpicture}
\caption{Evaluation of the three policies} \label{fig:utilityeq}
\end{figure}
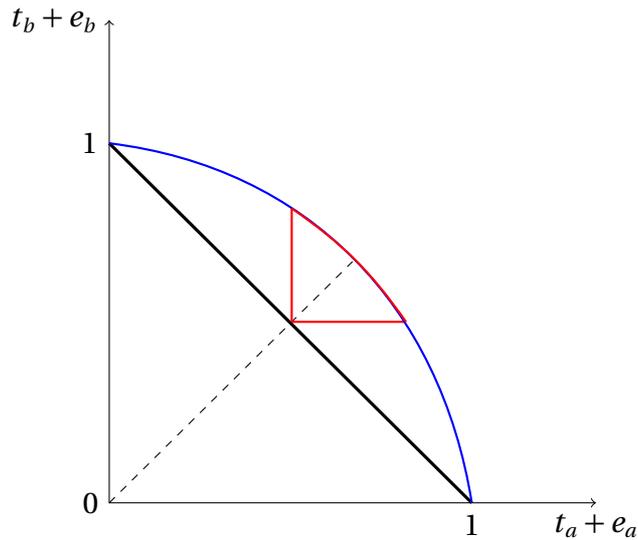
Each one-school policy entails investing in the education in one of the two sectors so as to move every individual from the distribution of talents to the blue curve, either horizontally or vertically.

However, the two-school policy would yield the same result, but at a lower cost; in fact, the red area would be saved. Hence, given that the  budget is fixed, the government could increase educational expenditures such that everyone reaches a strictly higher indifference curve. This result again generalizes to any pluralistic distribution of talents. In fact, we have the following theorem.
\begin{theorem}
    For any set of welfare weights, the two-school policy yields a higher social welfare than any of the one-school policies.
\end{theorem}

\section{Conclusions} \label{sec:conclusions}
In this paper, I tried to provide an answer to the following question : Is there a trade-off between EOp and \textit{Opportunity Pluralism}?

In pursuing this task, I tried to embed Fishkin's view about opportunities into a simple economic model in which an egalitarian government seeks to equalize opportunities among individuals whose preferences are uncertain at the time of educational investment, as well as endogenous to it. 

Under very mild assumptions, my analysis clearly goes against this claim. Moreover, the results seem to suggest that economic efficiency, both in terms of utility maximization and exhaustion of the available budget, \textit{implies} POp. In fact, an egalitarian government will always prefer a pluralistic opportunity structure (two-school society) over a unitary opportunity structure (one-school society).

Whereas the present analysis questions the fundamental ideas at the heart of \textit{Opportunity Pluralism}, it does not consider the new theory itself. The main idea is to define EOp in a \textit{negative} way; opportunities are equal among individuals when society is able to remove the different types of \textit{bottlenecks} which hinder the full realization of each person's ambitions, goals and attitudes. This is left for future research.

\bibliographystyle{apalike}
\bibliography{EOp}

\pagebreak

\appendix

\section{Proofs} \label{sec:Appendix}
\begin{proof}[Proof of \Cref{property:symmetry}]
By \Cref{assumption2} and the definition of $\mathcal V$ we have that
    \begin{align*}
        \mathcal{V}(w_{a}^i, w_{b}^i)&=\int \mathcal{P}(w_{a}^i, w_{b}^i;\theta) f(\theta; w_a^i, w_b^i) \ d\theta\\
        &= \int \max \big\{\theta \textit{w}_{a}^i , (1-\theta)\textit{w}_{b}^i\big\} f(\theta; w_a^i, w_b^i) \ d\theta\\
        &= \int \max \big\{(1-\theta)\textit{w}_{b}^i, \theta \textit{w}_{a}^i\big\} f(1-\theta;w_b^i,w_a^i) \ d\theta\\
        &= \int \mathcal{P}(w_{b}^i, w_{a}^i;1-\theta) f(1-\theta;w_b^i,w_a^i) \ d\theta\\
        &=\mathcal{V}(w_b^i,w_a^i) \ .
     \end{align*}
\end{proof}

\begin{proof}[Proof of \Cref{property:quasiconvexity}]
Consider two vectors, $(w_a^i, w_b^i)$, $(w_a^{'i}, w_b^{'i}) \in \Re_+^2$, and any convex combination between them, $(\Tilde{w}_a^i,\Tilde{w}_b^i) \in \Re_+^2$. 
Quasiconvexity requires that
\begin{equation} \label{eq:quasiconvexity}
    \max \big\{ \mathcal{V}(w_a^i, w_b^i),\mathcal{V}(w_a^{'i}, w_b^{'i}) \big\} \geq \mathcal{V}(\Tilde{w}_a^i,\Tilde{w}_b^i)
\end{equation}
for all $(w_a^i, w_b^i)$, $(w_a^{'i}, w_b^{'i})$, and $(\Tilde{w}_a^i,\Tilde{w}_b^i)$. We need to distinguish two cases.

\subsection*{Case I} If $w_a^i \geq (\leq) \ w_a^{'i}$ and $w_b^i \geq (\leq) \ w_b^{'i}$, then (\ref{eq:quasiconvexity}) follows from \Cref{property1}.
\subsection*{Case II} Consider now the case $w_a^i \geq  w_a^{'i}$ and $w_b^i \leq  w_b^{'i}$.
We have the following claim.
\begin{claim}
 $\max \big\{ \mathcal{V}(w_a^i, w_b^i),\mathcal{V}(w_a^{'i}, w_b^{'i}) \big\} \geq \mathcal{V}\big(\max\{w_a^i, w_a^{'i}\}, \max\{w_b^i, w_b^{'i}\}\big) \ .$ 
\end{claim}
\begin{proof} \ If  $f$ does not depend on wages, then 
\begin{align*}
    &\max \big\{ \mathcal{V}(w_a^i, w_b^i),\mathcal{V}(w_a^{'i}, w_b^{'i}) \big\} = \\
    &\max \bigg\{ \int \max \big\{\theta \textit{w}_{a}^i , (1-\theta)\textit{w}_{b}^i\big\} f(\theta; w_a^i, w_b^i) \ d\theta , \int \max \big\{\theta \textit{w}_{a}^{'i} , (1-\theta)\textit{w}_{b}^{'i}\big\} f(\theta; w_a^{'i}, w_b^{'i}) \ d\theta \bigg\} =\\
     &\max \bigg\{ \int \max \big\{\theta \textit{w}_{a}^i , (1-\theta)\textit{w}_{b}^i\big\} f(\theta) \ d\theta , \int \max \big\{\theta \textit{w}_{a}^{'i} , (1-\theta)\textit{w}_{b}^{'i}\big\} f(\theta) \ d\theta \bigg\} = \\
     &\int \max \bigg\{ \max \big\{\theta \textit{w}_{a}^i , (1-\theta)\textit{w}_{b}^i\big\}  , \max \big\{\theta \textit{w}_{a}^{'i} , (1-\theta)\textit{w}_{b}^{'i}\big\} \bigg\} f(\theta) \ d\theta = \\
     &\int \max \bigg\{\theta \max\{\textit{w}_{a}^i , w_a^{'i}\},  (1-\theta)\max\{\textit{w}_{b}^i, w_b^{'i} \} \bigg\} f(\theta) \ d\theta = \\
     &\mathcal{V}\big(\max\{w_a^i, w_a^{'i}\}, \max\{w_b^i, w_b^{'i}\}\big) \ .
\end{align*}
Consider now the case in which the two densities differ. 

Since $w_b^i \leq  w_b^{'i}$, whenever $\theta \leq \frac{w_b^{'i}}{w_a^i + w_b^{'i}}$, then $\max \big\{ \mathcal{V}(w_a^i, w_b^i),\mathcal{V}(w_a^{'i}, w_b^{'i}) \big\} = \mathcal{V}(w_a^{'i}, w_b^{'i})$. Moreover, since $w_a^i \geq  w_a^{'i}$, whenever $\theta \geq \frac{w_b^{'i}}{w_a^i + w_b^{'i}}$, then $\max \big\{ \mathcal{V}(w_a^i, w_b^i),\mathcal{V}(w_a^{'i}, w_b^{'i}) \big\} = \mathcal{V}(w_a^{i}, w_b^{i})$. In other words, the same reasoning applies inside the integrals, but now we are comparing values of $\theta$ from two different distributions.

However, by \Cref{assumption2}
\begin{equation*}
    1-F\bigg(\frac{w_b^{'i}}{w_a^i + w_b^{'i}}; \ w_a^i, w_b^i\bigg) \geq 1-F\bigg(\frac{w_b^{'i}}{w_a^i + w_b^{'i}}; \ w_a^{'i}, w_b^{'i}\bigg)
\end{equation*}
Therefore, the skewness of the densities makes the total mass of $\max \big\{ \mathcal{V}(w_a^i, w_b^i),\mathcal{V}(w_a^{'i}, w_b^{'i}) \big\}$ add up to more than 1. The claim then follows.
\end{proof}
Then, we have that
\begin{equation*}
    \max \big\{ \mathcal{V}(w_a^i, w_b^i),\mathcal{V}(w_a^{'i}, w_b^{'i}) \big\} \geq \mathcal{V}\big(\max\{w_a^i, w_a^{'i}\}, \max\{w_b^i, w_b^{'i}\}\big) \geq \mathcal{V}(\Tilde{w}_a^i,\Tilde{w}_b^i) \ ,
\end{equation*}
where the last inequality again follows from \Cref{property1}.
\end{proof}

\end{document}